\documentclass{article}
\usepackage[utf8]{inputenc}
\usepackage{graphicx}
\usepackage{mathtools,amsmath,amssymb,amsfonts,amsthm}
\usepackage{mathalfa}
\usepackage{tikz}
\usetikzlibrary{calc}
\usepackage{color}
\usepackage{url}
\usepackage{epstopdf}

\let\labelindent\relax
\usepackage{enumitem}

\graphicspath{{./figures/}}

\usepackage[english]{babel}

\usepackage{tabularx}

\usepackage{xcolor}

\usepackage{verbatim}
\usepackage{xspace}

\usepackage{amsfonts}
\usepackage{mathtools}
\usepackage{amsmath}

\usepackage{amssymb}

\usepackage{setspace}

\usepackage{fullpage}
\usepackage{ifthen}
\usepackage{tikz}
\usetikzlibrary{automata,positioning, arrows}

\theoremstyle{plain}
\newtheorem{thm}{Theorem}
\newtheorem{lemma}{Lemma}

\theoremstyle{definition}

\newtheorem{defn}{Definition}

\theoremstyle{remark}

\newtheorem{rem}{Remark}

\newcommand{\virgolette}[1]{``#1''}

\newcommand{\R}{\mathbb{R}}{}
{}
\newcommand{\N}{\mathbb{N}}

\newcommand{\bS}{\mathbb{S}}

\newcommand{\B}{\mathbb{B}}

\newcommand{\cA}{\mathcal{A}}
\newcommand{\cB}{\mathcal{B}}

\newcommand{\M}{\langle M \rangle}

\newcommand{\cF}{\mathcal{F}}

\newcommand{\wt}{\widetilde}

\DeclareMathOperator{\co}{co}
\DeclareMathOperator{\Inn}{Int}
\DeclarePairedDelimiterX{\inp}[2]{\langle}{\rangle}{#1, #2}

\title{Data-driven stability analysis of switched affine systems}
\author{Matteo Della Rossa \and Zheming Wang\and Lucas N. Egidio  \and Rapha\"el M. Jungers%
\thanks{RJ is a FNRS honorary Research Associate. This project has received funding from the European Research Council (ERC) under the European Union's Horizon 2020 research and innovation programme under grant agreement No 864017 - L2C. RJ is also supported by the Innoviris Foundation and the FNRS (Chist-Era Druid-net). The second and third authors contribute equally.}%
\thanks{M. Della Rossa, Z. Wang, L. N. Egidio and R. Jungers are with the ICTEAM,
        UCLouvain, 4 Av. G. Lema\^{i}tre, 1348 Louvain-la-Neuve, Belgium. ).
        {\tt\small \{matteo.dellarossa,zheming.wang,lucas.egidio,} {\tt\small raphael.jungers\}@uclouvain.be}}%
}

\newcommand{\lucas}[1]{{\color{green!50!black}#1}}
\newcommand{\rj}[1]{{\color{orange!80!black}#1}}
\newcommand{\matteo}[1]{{\color{red}#1}}
\newcommand{\zheming}[1]{{\leavevmode\color{blue}#1}}

\renewcommand{\lucas}[1]{#1}
\renewcommand{\rj}[1]{#1}
\renewcommand{\matteo}[1]{#1}
\renewcommand{\zheming}[1]{#1}

\begin{document}

\maketitle

\begin{abstract}
We consider discrete-time switching systems composed of a finite family of \emph{affine} sub-dynamics. \rj{First,} \lucas{we recall existing results and present further analysis on} the stability problem, the existence and characterization of compact attractors, and the relations these problems have with the joint spectral radius of the set of matrices composing the linear part of the subsystems. \rj{Second}, we tackle the problem of providing probabilistic certificates of stability along with the existence of forward invariant  sets, assuming no knowledge on the system data but only observing a finite number of sampled trajectories. Some numerical examples illustrate the advantages and limits of the proposed conditions. 

\end{abstract}

\section{Introduction}
The stability analysis of switching dynamical systems has been extensively studied in recent years, for an overview, see the survey~\cite{ShoWir07} or the book~\cite{Lib03}.
More precisely, starting with a finite family of vector fields $f_1,\dots,f_M:\R^n\to \R^n$, a switching system in discrete-time is described by the difference equation
\[
x(k+1)=f_{\sigma(k)}(x(k)),
\]
where $\sigma:\N\to \{1,\dots, M\}$ is the so-called switching signal. This signal can either be a control input to be designed or an external (and unpredictable) disturbance. In both these circumstances, it is often assumed that the sub-systems $f_1,\dots, f_M$ share a common equilibrium (w.l.o.g., the origin), and the stabilizability/stability problems are then studied with respect to this point.
Among many other examples, various techniques to tackle this and other related problems are proposed in ~\cite{AhmJunPar14},~\cite{GerCol06},~\cite{GoeHu06},~\cite{MolPya89} (and references therein), both in the linear or non-linear cases, and they mostly rely on the concept of \emph{common/multiple Lyapunov functions}. In the linear case (i.e., $f_i(x)\coloneqq A_ix$, for all $i\in \{1,\dots, M\},$ and  $A_i\in \R^{n \times n}$), asymptotic stability is equivalent to the condition that the joint spectral radius (JSR) $\rho(\cA)$ of the set $\cA\coloneqq\{A_1,\dots ,A_M\}$ is strictly less than $1$, see~\cite[Corollary 1.1]{Jung09}.

On the other hand, if the subsystems \emph{do not} share a common equilibrium, the stability analysis becomes more challenging and requires a different approach, since most of the aforementioned results can not be directly applied. The simplest case where this phenomenon appears is when the subsystems are \emph{affine} \lucas{maps}, leading to switching systems of the form
\begin{equation}\label{eq:AffineIntroduction}
x(k+1)=A_{\sigma(k)}x(k)+b_{\sigma(k)},
\end{equation}
where $\{A_1, \dots, A_M\}\subset \R^{n\times n}$, $\{b_1,\dots, b_M\}\subset \R^n$, and again $\sigma:\N\to \{1,\dots, M\}$ represents the switching signal.
This framework is \rj{of particular importance} since it can suitably provide a mathematical model for a large class of physical systems, specially in the power electronics domain, see~\cite{deaecto2010switched} for instance.

Several existing works related to System~\eqref{eq:AffineIntroduction} are devoted to the stabilizability problem \lucas{(i.e., $\sigma$ is a control input) and it is known in this case that asymptotic stability of a point can only be ensured if it is an equilibrium point  of one subsystem. However, other types of stability, such as practical stability or limit-cycle stability, can also be studied in this context to characterize more general behavior of the system trajectories in the steady state (see~\cite{deaecto2016stability,egidio2020global,egidio2021dynamic,sanchez2019practical,xu2010some}, and references therein).} On the other hand, the analysis of~\eqref{eq:AffineIntroduction} under arbitrary switching signals (and thus considering $\sigma$ as an \lucas{exogenous input}), is less studied in the literature. In this setting, a first step to characterize a minimal forward invariant set, under the hypothesis that $\rho(\cA)<1$, is provided in~\cite{AthJun16} and~\cite{KouRakKer05}. In the first part of this article, we \lucas{push further} this analysis, providing new results concerning forward invariance and attractiveness of sets for~\eqref{eq:AffineIntroduction}.

\rj{Recently}, particular attention from the scientific community was directed to the stability problem of (various kinds of) dynamical systems considered as \emph{black boxes}, see, e.g., ~\cite{INP:KQHT16}, ~\cite{KenBal19}, ~\cite{ART:WJ20}. More precisely, in this framework, the system data (in our case, the vector fields $f_1, \dots, f_M$) is completely or partially unknown, and the user only observes a finite number of its trajectories. Starting from this finite amount of information, some (probabilistic) certificates of stability and forward invariance can still be obtained, see, for example~\cite{KenBal19},~\cite{Kor20},~\cite{ART:WJ20,INP:RWJ21,ART:BJW21}.

In this article, we \lucas{generalize} these ideas to affine switched systems as~\eqref{eq:AffineIntroduction}, establishing both probabilistic upper bound for the JSR $\rho(\cA)$ and probabilistic certificates of forward invariance. The main ideas behind our construction are inspired by the recent developments provided in~\cite{KenBal19}, which strictly considers the linear case.  As previously mentioned, some particular care is required in the affine setting, due to the absence, in general, of a common equilibrium. We provide two different methods; the first one requires an \textit{a priori} knowledge of \rj{an upper bound on} the norms of the vectors $b_1,\dots, b_M$, while the second approach relaxes this assumption, leading to an alternative stability certificate. The theoretical developments, their advantages and limits are finally illustrated with the aid of a numerical example.

\textbf{Notation:} The notation used through this article is standard. In particular we define  $\R_{\geq 0}\coloneqq\{x\in \R\,\,\,\vert\,x\geq 0\}$. Given any $M\in \N$ we define $\M\coloneqq\{1, \dots, M\}$. Given a set of matrices $\cA=\{A_1, \dots, A_M\}\subset \R^{n\times n}$, we denote by $\rho(\cA)$ the joint spectral radius of $\cA$, see \lucas{Appendix~\ref{app:jsr} for the formal definition and further discussions}. With $|\cdot|:\R^n\to \R_{\geq 0}$ we denote the standard Euclidean norm in $\R^n$. A symmetric and positive definite matrix $P\in \R^{n\times n}$ is denoted by $P\succ0$; we denote its \emph{condition number} by $\kappa(P):=\frac{\lambda_{\max}(P)}{\lambda_{\min}(P)}$ and the norm associated to $P$ by $\|x\|_P\coloneqq\sqrt{x^\top P x}$..
Given a set $C\subseteq\R^n$, $\Inn(C)$, $\partial C$, $\co(C)$ denote the interior, boundary and convex hull of $C$, respectively. For any $n\in \N$, $\Lambda^n$ denotes the standard probability simplex, i.e. $\Lambda^n\coloneqq\{\lambda\in \R^n\;\vert\;\sum_{j=1}^n\lambda_j=1,\;\;\lambda_j\geq 0,\;\forall j\in \langle n\rangle\}$.

\section{STABILITY ANALYSIS PRELIMINARIES}
For a given $M\in \N$, consider \lucas{$\cF\coloneqq\{ (A_i, b_i)\vert\, i\in \M\}$ a set of pairs of matrices  $A_i\in \R^{n \times n}, b_i\in \R^n,~\forall i\in\M$} characterizing the dynamics of subsystems for a given discrete-time switching system as
\begin{equation}\label{eq:Switchingsystems}
x(k+1)=A_{\sigma(k)}x(k)+b_{\sigma(k)},\quad x(0)=x_0,
\end{equation}
where $x:\N\to\R^n$ is the state \matteo{signal} and  $\sigma:\N\to \M$ is an arbitrary switching signal. We denote $\cA\coloneqq\{A_1, \dots, A_M\}\subset \R^{n\times n}$ and $\cB\coloneqq\{b_1,\dots, b_M\}\subset \R^n$ to refer to the set of dynamic matrices and affine terms, respectively.

Equivalently, we consider  the difference inclusion
\begin{equation}\label{eq:DiffInc}
x^+\in F(x)\coloneqq\{A_i x+b_i,\,\,i \in \M\}
\end{equation}
that holds pointwise in time for any trajectory obtained from~\eqref{eq:Switchingsystems}.
Given a switching sequence $\sigma:\N \to \M$ and a $j\in \N$, we define $A^j_{\sigma}\coloneqq A_{\sigma(j-1)}A_{\sigma(j-2)}\cdot\cdot\cdots A_{\sigma(0)}$, and by convention $A^0_{\sigma}=I$. Therefore, for an arbitrary initial condition $x_0\in \R^n$ the state trajectory of~\eqref{eq:Switchingsystems} under  $\sigma$ is 
\begin{equation}\label{eq:solution}
x(k,\sigma,x_0)=A^k_{\sigma}x_0+\sum_{j=0}^{k-1}A^j_{\sigma}b_{\sigma(k-j-1)}.
\end{equation}
It is noteworthy that the map to $k$ steps ahead  generated by~\eqref{eq:Switchingsystems} can also be represented as a switched affine system defined by the solution~\eqref{eq:solution}.

\matteo{
After some preliminary definitions, we recall results ensuring that if the joint spectral radius  $\rho(\cA)<1$ (see \lucas{Appendix~\ref{app:jsr}} for the formal definition), then System~\eqref{eq:Switchingsystems} admits a (non-empty) forward invariant and attractive compact set.
}
\begin{defn}
Given $C\subseteq \R^n$ a compact set and a difference inclusion as in~\eqref{eq:DiffInc}, we say that
\begin{enumerate}
    \item $C$ is \emph{forward invariant} if $F(C)\subseteq C$ and $C$ is \emph{strictly forward invariant} if $F(C)\subseteq \Inn(C)$,
    \item A forward invariant $C$  is \emph{minimal} if, for every forward invariant set $D$, it holds that $C\subseteq D$,
    \item $C$ is attractive if for every $x_0\in \R^n$ and every $\sigma:\N\to \M$, we have
    \[
    \text{dist}( x(k,\sigma, x_0), C)\to 0,\,\,\,\,\text{as } \,k\to +\infty.
    \]
\end{enumerate}
\end{defn}
\begin{rem}
Note that if a minimal forward invariant set $C$ exists, then it is unique and, moreover, it has the property that $F(C)=C$. Indeed,
$F(F(C))\subseteq F(C)$, and thus $F(C)$ is also forward invariant; from minimality of $C$, we thus have $C\subseteq F(C)$. Hence, the minimal forward invariant set can be interpreted as a \virgolette{set-valued equilibrium} for~\eqref{eq:DiffInc}.\hfill $\triangle$
\end{rem}
\matteo{We now introduce a property of affine switched systems that is used in what follows.}
\begin{lemma}\label{lemma:ConvexClosure}
Given any set $C\subseteq \R^n$ \matteo{and a set-valued map $F:\R^n \rightrightarrows \R^n$   as in~\eqref{eq:DiffInc}}, we have $F(\co(C))\subseteq\co(F(C))$.
In particular, this implies that if $C$ is forward invariant for~\eqref{eq:Switchingsystems}, then $\co(C)$ is also forward invariant.
\end{lemma}
\begin{proof}
Consider an arbitrary $y\in F(\co(C))$. By the Caratheodory Theorem, there exist $x_1, \dots, x_{n+1}\in S$ and $\lambda \in \Lambda^{n+1}$ such that $y\in F(\sum_{j=1}^{n+1}\lambda_j x_j)$. Therefore, there exists an $i\in \M$ such that 
\[
y= A_i\left(\sum_{j=1}^{n+1}\lambda_j x_j\right)+b_i=\sum_{j=1}^{n+1}\lambda_j(A_ix_j+b_i)\in \co\{F(C)\}.
\]
For the second part, consider $C$ forward invariant (i.e. $F(C)\subseteq C$) then
\[
F(\co(C))\subseteq\co(F(C))\subseteq \co(C),
\]
proving that $\co(C)$ is forward invariant.
\end{proof}
Note that, in general, this lemma does not hold for non-affine vector fields.
 
 \matteo{
 In what follows, assuming $\rho(\cA)<1$, we consider a $\wt \rho\in \R$ such that $\rho(\cA)<\wt \rho<1$ and we denote by $\|\cdot\|_\cA:\R^n\to \R_{\geq 0}$ a \emph{$\wt \rho$-contractive} norm, i.e. a norm satisfying 
 \begin{equation}\label{eq:LyapNorm}
     \|Ax\|_\cA\leq \wt \rho\|x\|_\cA,\,\,\,\forall\,x\in \R^n,\,\,\forall \,A\in \cA.
 \end{equation}
 The existence of this norm is ensured by the result recalled in Lemma \ref{lemma:ExtremalNorms} in Appendix \ref{app:jsr}; with $\B_\cA$ we denote the unit ball of $\|\cdot\|_\cA$, i.e. $\B_\cA:=\{x\in \R^n\,\vert\, \|x\|_\cA\leq 1\}$ .}
 
At this point, we can recall the main stability result for~\eqref{eq:Switchingsystems}, derived from the discussions presented in~\cite{AthJun16}.

\begin{lemma}\label{lemma:Kinfty}
If $\rho(\cA)<1$ there exists a non-empty, compact, minimal and attractive forward invariant set $K_\infty$ for~\eqref{eq:Switchingsystems}. In addition,
\begin{equation}\label{eq:upperBound}
K_\infty \subseteq\frac{\|\cB\|_\cA}{1-\wt \rho}\B_\cA,
\end{equation}
where $\|\cB\|_\cA\coloneqq\max_{i\in \M}\{\|b_i\|_\cA\}$. 
\end{lemma}
The proof of this lemma can be found in~\cite{AthJun16}, but, for completeness, we recall here the construction of $K_\infty$.
The expression~\eqref{eq:solution} suggests the following iterative definition:
\begin{equation}
\begin{aligned}
K_0&\coloneqq\{0\},\\
K_1&\coloneqq\cB=\{b_1,\dots, b_m\},\\
&\vdots\\
K_k&\coloneqq\{A_ix+b_i\,\,\vert i\in \M, \,x\in K_{k-1}\},
\end{aligned}
\end{equation}
for any $k\in \N$.
Equivalently, for every $k>0$,
\begin{equation}\label{eq:DefSet} 
K_k=\left\{ \sum_{j=0}^{k-1}A^j_{\sigma}b_{\sigma(k-j)}\;\vert\;\;\sigma:\langle k\rangle\to\M\right\}
\end{equation}
We then define $K_\infty\coloneqq\lim_{k\to \infty}K_k$, and following~\cite{AthJun16}, it can be shown that $K_\infty$ is well-defined and it satisfies the properties in Lemma~\ref{lemma:Kinfty}.
Note that, in the particular case where $b_i=0,~\forall i\in \M$, we obtain $K_k=\{0\}$ for any $k\in \N$, and thus $K_\infty=\{0\}$, as expected. 
Moreover, defining $c_i\coloneqq(I-A_i)^{-1}b_i$ for $i\in \M$, it can be seen that $\{c_i\,\,\vert\,\,i\in \M\}\subseteq K_\infty$. This can be done by evaluating~\eqref{eq:solution} for constant  switching signals with $k\rightarrow\infty$, and recalling that, for any Schur stable matrix $A$, the so-called Neumann series yields $(I-A)^{-1}=\sum_{j=0}^\infty A^j$.
Finally, notice that $\co(K_\infty)$ is the minimal \emph{convex} forward invariant set, and again we have 
$
\co(K_\infty) \subseteq\frac{\|\cB\|_\cA}{1-\wt \rho}\B_\cA
$,
since the unit ball of any norm is convex.

Now that we have theoretically characterized the minimal forward invariant and attractive set, we present some crucial results for the data-driven stability analysis, which will be introduced in the subsequent section. For the sake of readability, in the following lemmas we consider a single affine map, \matteo{but they can be generalized  for the set-valued map defined in~\eqref{eq:DiffInc}.}
\begin{lemma}\label{lemma:SublevelSetLyap}
Consider an affine function $f:\R^n\to \R^n$, and a real-valued function $v:\R^n\to \R$ convex, bounded from below and radially unbounded. Consider $K\coloneqq\{x\in \R^n\,\vert\,\, v(x)\leq R\}$, for some $R\in \R$.
If
\[
v(f(x))\leq v(x) ,\,\,\forall x\in \partial K,
\]
then $K$ is a compact and convex forward invariant set for the system $x^+=f(x)$.
\end{lemma}
\begin{proof}
Compactness and convexity of $K$ trivially follow from the definition of $v$. Regarding the forward invariance, note that
\[
\max_{x\in K} v(f(x))=\max_{x\in \partial K}v(f(x))\leq \max_{x\in \partial K} v(x)\leq R,
\]
since $v\circ f:\R^n\to \R$ is convex and, hence, its maximum over a convex set $K$ is attained on the boundary $\partial K$,
(\cite{RockConv}).
\end{proof}
We say that a set $X\subset \R^n$ is a \emph{symmetric $C$-set} if $X$ is compact and convex, $0\in \Inn(X)$ and $x\in X$ implies $-x\in X$. \matteo{In the following result, we need to strengthen the hypothesis of Lemma \ref{lemma:SublevelSetLyap}, considering norms induced by inner products, which, on $\R^n$, is equivalent to consider functions $v:\R^n\to \R$ defined by $v(x):=\sqrt{x^\top Px}$ for some positive definite matrix $P\in \R^{n\times n}$, $P>0$. Note that any norm in particular satisfies the hypothesis of Lemma \ref{lemma:SublevelSetLyap}.}
\begin{lemma}\label{lemma:CoreLemma}
Consider a norm $v:\R^n\to \R$ \matteo{induced by an inner product}, and an affine map $f:\R^n\to \R^n$ defined by $f(x)\coloneqq Ax+b$. Suppose  that
\[
v(Ax+b)\leq \rho v(x),\,\,\,\forall \,\,x\in \partial X,
\]
where $X$ is a symmetric $C$-set and $\rho\in \R_{\geq 0}$.
Then 
\begin{enumerate}
\item $v(Ax)\leq \rho v(x)$, $\forall \,\,x\in \R^n$,
\item $v(Ax+b)\leq \rho v(x)$ $\forall \,\,x\in \R^n\setminus X$,
\item If $\rho<1$, the set $K\coloneqq\{x\in \R^n\,\vert v(x)\leq \max_{y\in \partial X}v(y)\}$ is attractive and forward invariant for $x^+=Ax+b$.\end{enumerate}
\end{lemma}
\begin{proof}
For Item 1), note that the case $x=0$ is trivial. Consider a generic $x\in \R^n\setminus \{0\}$, and suppose by contradiction that $v(Ax)> \rho v(x)$. By linearity, the same holds for every $y\in \text{span}(x)\setminus \{0\}$. Consider $\hat x\in \text{span}(x)\cap \partial X$, by symmetry of $X$, also $-\hat x\in \partial X$.
By the parallelogram law we have
\[
\begin{aligned}
v^2(A\hat x+b)+v^2(-A\hat x+b)&= 2v^2(A\hat x)+2v^2(b)\\
&> 2\rho^2v^2(\hat x)+2v^2(b).
\end{aligned}
\] 
But, by hypothesis, $\hat x,-\hat x\in \partial X$, and then 
\[
v^2(A\hat x+b)+v^2(-A\hat x+b)\leq \rho^2v^2(\hat x)+\rho^2v^2(-\hat x)=2\rho^2v^2(\hat x),
\]
leading to a contradiction.
For Item 2) consider a generic $x\in \R^n\setminus X$, we can write $x=\lambda \hat x$, with $\lambda>1$ and $\hat x\in \partial X$. Computing
\[
\begin{aligned}
v(Ax+b)&=v(\lambda A \hat x+b)=v((\lambda-1)A\hat x+A\hat x+b)\\&\leq v((\lambda-1)A\hat x)+v(A\hat x+b)\\&=(\lambda-1)v(A\hat x)+v(A\hat x+b)\\&\leq (\lambda-1)\rho v(\hat x)+\rho v(\hat x)=\rho \lambda v(\hat x)=\rho v(x),
\end{aligned}
\]
where we have applied Item 1).

For Item 3), since $v$ is convex, note that $X\subseteq K$,  proving attractiveness, since $v(A_ix+b_i)-v(x)<0$, for all $x\notin K$, for all $i\in \M$. Moreover, $\partial K\subset \R^n\setminus \Inn(X)$, which ensures by Item 2) that, for any $y\in \partial K$,  we have
\[
 v(Ay+b)\leq v(y).
\]
Hence, Lemma~\ref{lemma:SublevelSetLyap} provides the forward invariance property and concludes the proof.
\end{proof}

\section{DATA-DRIVEN STABILITY ANALYSIS}
\label{sec:datadriven}
\matteo{
In this section, in analyzing~\eqref{eq:Switchingsystems}, we suppose that we \emph{do not} have access to the system dynamics $\cF\coloneqq\{ (A_i, b_i)\vert\, i\in \M\}$; the only information available is the number of subsystems $M\in \N$. On the other hand, we suppose that we observe $N$ trajectories of length $1$ of system \eqref{eq:Switchingsystems}, denoted, for every $i \in  \langle N \rangle$, by $(x_{0,i},x_{1,i})$, with  $x_{1,i}=A_{j_i}x_{0,i}+b_{j_i}$ for some $j_i\in \M$. The initial points  $x_{0,i}$ will be chosen on a sphere of radius $R>0$ denoted by $\bS_R:=\{x\in \R^n\,\vert\,|x|=R\}$; we assume that  the user can choose the radius $R$ depending on the particular setting, (further discussions on this choice are provided in the sequel). We suppose that the initial points $x_{0,i}\in \bS_R$ and the modes $j_i\in \M$ are sampled with respect to uniform probability measure on $\bS_R$ and $\M$, respectively. We refer to Appendix \ref{app:prob} for the formal definition of these probability measures and related concepts.

Our goal is, starting with this limited information, to provide a probabilistic certificate of stability for~\eqref{eq:Switchingsystems}  and to provide a bounded set, together with a probabilistic certificate of forward invariance. This will be done by generalizing, for switched affine systems, the results presented in~\cite{KenBal19}, which tackle exclusively linear mappings of the form $x^+\in \{A_ix\;\vert\;i\in \M\}$. Our results, however, are not a trivial byproduct of these recent developments, due to the complex nature of the attractor (Lemma~\ref{lemma:Kinfty}) and  the lack of homogeneity of system~\eqref{eq:Switchingsystems} (in general, $F(\lambda x)\neq \lambda F(x)$ for $\lambda\in \R$, for $F$ defined in~\eqref{eq:DiffInc}).
}
\subsection{Scenario optimization}
We first introduce a data-driven Lyapunov-like stability problem in the spirit of the scenario approach~\cite{ART:C10}. For a fixed $R>0$, let us define $Z: = \bS_R \times \M$; from now on, we consider the uniform probability measure $\mu$ on $Z$, \matteo{see Appendix \ref{app:prob} for the formal definition.} Given a data set of $N$ samples, $\omega_N\coloneqq\{(x_i,\sigma_i) \in \bS_R \times \M: i = 1,2,\cdots, N \}$, we define the following \emph{sampled problem}:
\begin{subequations}\label{eqn:mathcalPomegaN}
\begin{align}
\mathcal{P}(\omega_N): \quad  &\gamma(\omega_N): = \min_{\gamma\ge 0, P} \gamma\\
\textrm{s.t.} \quad  & (A_\sigma x+b_{\sigma})^\top P (A_\sigma x+b_{\sigma}) \le \gamma^2 x^\top Px, \nonumber\\
&\qquad  \forall (x,\sigma) \in \omega_N, \label{eqn:Asigmaxbsigma}\\
& P \succ 0, P = P^\top.\label{eqn:P_geq_0}
\end{align}
\end{subequations}
\matteo{Note that Problem \eqref{eqn:mathcalPomegaN} could have multiple solutions, due to the fact that the objective function does not explicitly depend on $P$. On the other hand, from now on, we suppose to have a \virgolette{tie-breaking rule} which allows us to select one particular solution. For more general discussion related to uniqueness of solutions and tie-breaking criteria we refer to \cite{campi2018general}. The particular rule considered in this paper is introduced in what follows (Problem \eqref{eqn:second_step_optim_prob}).
Thus, from now on, with an abuse of notation, we denote with $(\gamma(\omega_N),P(\omega_N))$ \emph{the} solution to Problem~(\ref{eqn:mathcalPomegaN}) (chosen by the particular tie-breaking rule).} In order to derive probabilistic guarantees, we recall the definition of \emph{support subsamples} in \cite[Definition 2]{campi2018general}.
\zheming{
\begin{defn}\label{def:somega}
A subsample $\omega'_N \subseteq \omega_N$ is a \emph{support subsample} of $\mathcal{P}(\omega_N)$ if $(\gamma(\omega'_N),P(\omega'_N)) = (\gamma(\omega_N),P(\omega_N))$. 
\end{defn}
We then define:
\begin{align}
    s(\omega_N):= & \min_{\omega'_N \subseteq \omega_N} |\omega'_N| \nonumber \\
    \textrm{ s.t. } & \omega'_N \textrm{ is a support subsample} \label{eqn:somegaN}
\end{align}
where $|\omega'_N|$ is the cardinality of $\omega'_N$.
}

At this point, we adapt the chance-constrained Theorem in~\cite{ART:C10} to our stability problem as follow.

\zheming{
\begin{thm}[Theorem 1 of~\cite{campi2018general}]\label{thm:scenario}
For any given $N\in \N$, let $\omega_N$ be independent and identically distributed (i.i.d) with respect to the probability measure $\mu$ on $Z$. Consider the problem $\mathcal{P}(\omega_N)$ given in~(\ref{eqn:mathcalPomegaN}) with the solution being $(\gamma(\omega_N),P(\omega_N))$. Then, for any $\beta\in (0,1)$, the following holds:
\begin{align}
    \mu^{N}\{ \omega_N \in Z^N: \mu(V(\omega_N)) > \varepsilon(s(\omega_N)) \} \le \beta
\end{align}
where $V(\omega_N) \coloneqq \{(x,\sigma)\in Z: \|A_\sigma x+b_{\sigma}\|_{P(\omega_N)} > \gamma(\omega_N) \|x\|_{P(\omega_N)}\}$, $s(\omega_N)$ defined in (\ref{eqn:somegaN}), and $\varepsilon: \{0,1,\cdots, N\}\rightarrow [0,1]$ is a function given as:
\begin{align}\label{eqn:varepsilonk}
\varepsilon(k):=\begin{cases} 
      1 & \textrm{if } k=N; \\
      1-\sqrt[N-k]{\frac{\beta}{N {N \choose k}}} & 0\le k<N.
   \end{cases}
\end{align} 
\end{thm}
}


\zheming{
\begin{rem}
As shown in \cite{ART:BJW21}, with a tie-breaking rule and an additional constraint $\|P\|_F \le C$ for some large $C>0$ ($\|P\|_F$ denotes the Frobenius norm of $P$), the quantity $s(\omega_N)$ as given in (\ref{eqn:somegaN}) is always bounded from above by $d = \frac{n(n+1)}{2}$, which allows to have a better bound in Theorem \ref{thm:scenario} (see Remark 4 of \cite{campi2018general}) by choosing
\begin{align}\label{eqn:varepsilonkd}
\varepsilon(k):=\begin{cases}
      1 & \textrm{if } k\ge d+1; \\
      1-\sqrt[N-k]{\frac{\beta}{(d+1) {N \choose k}}} & \textrm{otherwise}.
   \end{cases}
\end{align} 
\end{rem}
In some situations solving Problem (\ref{eqn:mathcalPomegaN}) may lead to an ill-conditioned $P(\omega_N)$. This may lead to  deteriorate the quality of the bounds that we will obtain below. To overcome this issue, we adopt the following tie-breaking rule:
\begin{subequations}\label{eqn:second_step_optim_prob}
\begin{align}
\mathcal{Q}(\omega_N,\gamma): \quad  & P(\omega_N,\gamma): = \arg\displaystyle\min_{P,\alpha} \alpha + c\|P\|^2_F\\
\textrm{s.t.}\quad & \text{\eqref{eqn:Asigmaxbsigma}-\eqref{eqn:P_geq_0}}\\
& \alpha I\succeq P \succeq I
\end{align}
\end{subequations}
where $c>0$ is a weighting parameter.} Also, notice that problem $\mathcal{P}(\omega_N)$ in~\eqref{eqn:mathcalPomegaN} reduces to a generalized eigenvalue problem (GEVP), which can be efficiently solved by off-the-shelf tools (see~\cite{boyd1994linear}).

Departing from Theorem \ref{thm:scenario}, we provide two probabilistic stability certificates, each one with its own advantages/drawbacks.
More specifically:
\begin{itemize}
    \item The first stability certificate, presented in Subsection~\ref{subsec:FirstMethod}, requires an a-priori deterministic uniform upper bound on the norms of the vectors $b_1,\dots, b_M$. The estimate then takes into account this information, providing an upper bound on the joint spectral radius of $\cA$ which is eventually tight, if the user has the freedom to sample initial conditions from a sphere with arbitrary large radius $R>0$,
    \item The second one, presented in Subsection~\ref{subsec:SecondMethod},  makes use of Lemma~\ref{lemma:CoreLemma}. It does not require any \textit{a priori} information on the system but, on the other hand, its probabilistic estimation of the JSR of $\cA$ will be affected by the condition number of $P(\omega_N)\succ0$, solution to~\eqref{eqn:mathcalPomegaN}.
\end{itemize}

\begin{rem}[Trajectories of arbitrary length]
Our approach can be generalized considering trajectories of arbitrary length $l\in N$, i.e. observing $N$ pairs $(x_{0,i},x_{l,i})_{1\leq i\leq N}$, with $x_{0,i}\in \bS_R$ and $x_{l,i}$ the state of the solution of~\eqref{eq:Switchingsystems} after $l$ steps, corresponding to a randomly sampled sequence of modes $(j_1,\dots,j_l)\in \M^l$. This, recalling~\eqref{eq:solution}, allows us to provide probabilistic upper bound on the joint spectral radius of the set $\cA^l\coloneqq\{A_{j_1}\cdot\cdot\cdot A_{j_l}\,\vert\,A_{j_i}\in \cA\}$, i.e. the set all the products of length $l$ of matrices in $\cA$. Since $\rho(\cA^l)=\rho(\cA)^l$ (see~\cite[Section 1.2.1]{Jung09}) this approach will improve our estimation of the joint spectral radius of $\cA$, as done for example in~\cite{KenBal19}. We decided to present our results for trajectories of length $1$, in order to keep the notation and the developments less convoluted. \hfill $\triangle$
\end{rem}

\subsection{First stability certificate}\label{subsec:FirstMethod}
Suppose that an upper bound on $\max_{i\in \M}|b_i|$ is available and let it be denoted as $B$, i.e., $B\ge \max_{i\in \M}|b_i|$. With this upper bound, the following lemma can be derived from Theorem~\ref{thm:scenario}.
\zheming{
\begin{lemma}\label{lem:Vs}
For any given $N\in \N$, let $\omega_N$ be i.i.d with respect to the probability measure $\mu$ on $Z$. Given the solution $(\gamma(\omega_N),P(\omega_N))$ of $\mathcal{P}(\omega_N)$ defined as in~(\ref{eqn:mathcalPomegaN}), let us define:
\begin{align}
    \overline{\gamma}(\omega_N) \coloneqq \gamma(\omega_N) +  \frac{B}{R} \sqrt{\kappa(P(\omega_N))}
\end{align}
where $B \ge \max_{i\in \M}\|b_i\|$. Then, for any $\beta\in (0,1)$, the following holds:
\begin{align}
    \mu^{N}\{ \omega_N \in Z^N: \mu(V^*(\omega_N)) > \varepsilon(s(\omega_N)) \} \le \beta
\end{align}
where $V^*(\omega_N) \coloneqq \{(x,\sigma)\in Z: \|A_\sigma x\|_{P(\omega_N)} > \overline{\gamma}(\omega_N) \|x\|_{P(\omega_N)}\}$, $s(\omega_N)$ is defined in (\ref{eqn:somegaN}), and $\varepsilon: \{0,1,\cdots, N\}\rightarrow [0,1]$ is a function defined as in (\ref{eqn:varepsilonk}).
\end{lemma}
}
\begin{proof}
First, we show that $V^*(\omega_N) \subseteq V(\omega_N)$ for any $\omega_N$. Consider any $(x,\sigma)\in V^*(\omega_N)$, we have
\begin{align*}
    &\|A_\sigma x+b_{\sigma}-b_{\sigma}\|_{P(\omega_N)} > \overline{\gamma}(\omega_N)\|x\|_{P(\omega_N)} \Rightarrow  \\
    &\|A_\sigma x+b_{\sigma}\|_{P(\omega_N)} +B\sqrt{\lambda_{\max}(P(\omega_N))}\\
    &> \left( \gamma(\omega_N) +  \frac{B}{R} \sqrt{\kappa(P(\omega_N))} \right)\|x\|_{P(\omega_N)} \Rightarrow  \\
    &\|A_\sigma x+b_{\sigma}\|_{P(\omega_N)} +B\sqrt{\lambda_{\max}(P(\omega_N))}\\
    &> \gamma(\omega_N)\|x\|_{P(\omega_N)}  + \frac{B}{R} \sqrt{\kappa(P(\omega_N))} R \sqrt{\lambda_{\min}(P(\omega_N))}\\
    &\Rightarrow \|A_\sigma x+b_{\sigma}\|_{P(\omega_N)} > \gamma(\omega_N)\|x\|_{P(\omega_N)}
\end{align*}
This means that $(x,\sigma)\in V(\omega_N)$ from the definition of $V(\omega_N)$. Thus, $V^*(\omega_N) \subseteq V(\omega_N)$. As a result, $\mu(V^*(\omega_N)) > \varepsilon$ implies $\mu(V(\omega_N)) > \varepsilon$. Then, from Theorem~\ref{thm:scenario}, the statement holds.
\end{proof}
Before proceeding, we need to introduce the following notation. Consider any subset of the sphere $\wt S\subset \bS_R$ with fixed spherical measure $\sigma^{n-1}(\wt S)=\varepsilon\leq 1$ (see Appendix~\ref{app:prob} for the definition); we want to characterize the radius of the biggest ball contained in $\co(\bS_R\setminus\wt S)$. More precisely, we define the function $\delta:[0,1]\to [0,1]$ as
\begin{equation}\label{eq:FunctionDelta}
\delta(\varepsilon)\coloneqq\sup\left\{s\leq 1\; \Bigg\vert\;\begin{aligned} &\bS_{sR}\subset \co(\bS_R\setminus \wt S),\\&\forall \wt S\subset\bS_R,\;\text{s.t. }\sigma^{n-1}(\wt S)=\varepsilon\end{aligned}\right\}.
\end{equation}
This function $\delta$ will be used in what follows, and it is studied, both from a theoretical and numerical point of view, in~\cite[Proposition 13]{KenBal19}, to which we refer. It is clear that $\delta$ is strictly decreasing, $\delta(0)=1$ and $\delta(1)=0$.

Lemma \ref{lem:Vs} allows us to apply Theorem 15 in~\cite{KenBal19} to affine systems and derive a stability guarantee, as shown below.
\zheming{
\begin{thm}\label{thm:first_bound}
For any given $N\in \N$, let $\omega_N$ be i.i.d with respect to the probability measure $\mu$ on $Z$. Then, for any $\beta \in (0,1)$, we have the following:
\begin{align}
\mu^{N}\left\{ \omega_N \in Z^N: \rho(\cA) \le \bar\rho_1(\omega_N) \right\} \ge 1- \beta
\end{align}
where
\begin{equation}\label{eqn:rhobar1}
    \bar\rho_1(\omega_N) = \frac{\gamma(\omega_N) +  \frac{B}{R} \sqrt{\kappa(P(\omega_N))}}{\sqrt{\delta(M\overline{\kappa}(P(\omega_N))\varepsilon(s(\omega_N)))} },
\end{equation}
$s(\omega_N)$ defined in (\ref{eqn:somegaN}), $\varepsilon: \{0,1,\cdots, N\}\rightarrow [0,1]$ is a function defined as in (\ref{eqn:varepsilonk}), and $\overline{\kappa}(P)\coloneqq\sqrt{\frac{\det(P)}{\lambda_{\min}(P)^n}}$.
\end{thm}
}
\begin{proof}
This result is a direct consequence of Theorem 15 in~\cite{KenBal19} and Lemma~\ref{lem:Vs}.
\end{proof}
From Theorem \ref{thm:first_bound}, we observe that the bound in (\ref{eqn:rhobar1}) converges to the one in the linear case in \cite[Theorem 15]{KenBal19} as $R$ goes to infinity.

\subsection{Second Stability Certificate}\label{subsec:SecondMethod}
In this subsection we propose an alternative method to provide probabilistic upper bound of the joint spectral radius $\rho(\cA)$ starting from a solution of~\eqref{eqn:mathcalPomegaN}, but without requiring any \textit{a priori} upper bound on the norms of $b_1,\dots, b_M$.

Again, suppose we have a solution $(\gamma,P)\coloneqq(\gamma(\omega_N), P(\omega_N))$ to~\eqref{eqn:mathcalPomegaN}. 
\zheming{
Rephrasing Theorem~\ref{thm:scenario}, for any fixed $\beta>0$, with probability no smaller than $1-\beta$, we have
\begin{equation}\label{eq:almostInvaraince}
\|A_ix+b_i\|_P\leq  \gamma \|x\|_P\,\,\,\forall x\in \bS_R\setminus \wt S,\,\,\,\forall \,i\in \M,
\end{equation}
for some $\wt S\subset \bS_R$ with spherical measure $\sigma^{n-1}(\wt S)\leq \varepsilon(s(\omega_N)) M$, see~\cite[Corollary 11]{KenBal19}.
}
From the probabilistic \virgolette{almost invariance} presented in~\eqref{eq:almostInvaraince}, we want to derive an upper bound for $\rho(\cA)$, we thus need the following result.

\begin{lemma}\label{lemma:Convexity}
Suppose that~\eqref{eq:almostInvaraince} holds for a certain set $\wt S\subset \bS_R$. Let us define $L\coloneqq\max_{y\in \bS_R}\|y\|_{P}=\sqrt{\lambda_{max}(P)}R$. It holds that
\[
\begin{aligned}
\|A_ix+b_i\|_P\leq  \gamma L\,\,\,\forall x\in \co(\bS_R\setminus \wt S),\,\,\,\forall \,i\in \M.
\end{aligned}
\]
\end{lemma}
\begin{proof}
 Consider $x\in \co(\bS_R\setminus \wt S)$, we can write $x=\sum_{j=1}^{n+1} \lambda_j x_j$, for some $\lambda\in \Lambda^{n+1}$ and some $x_1,\dots, x_{n+1}\in \bS_R\setminus \wt S$. Considering any $i\in \M$, we have 
\[
\begin{aligned}
\|A_ix&+b_i\|_P\hspace{-0.1cm}=\hspace{-0.1cm}\left\|\sum_{j=1}^{n+1} A_i\lambda_jx_j+b_i\right\|_P\hspace{-0.1cm}\!\!=\left\|\sum_{j=1}^{n+1}\lambda_j (A_ix_j+b_i)\right\|_P\\&\leq \sum_{j=1}^{n+1}\lambda_j \|A_ix_j+b_i\|_P\leq \sum_{j=1}^{n+1}\lambda_j \gamma\|x_j\|_P\leq  \gamma L,
\end{aligned}
\]
concluding the proof.
\end{proof}
Now consider $r=R\delta(\varepsilon M)$, where $\delta:[0,1]\to[0,1]$ is defined in~\eqref{eq:FunctionDelta}.
Since $\bS_r\subset \co(\bS_R\setminus \wt S)$, by Lemma~\ref{lemma:Convexity} we have
\[
\|A_ix+b_i\|_P\leq \gamma L= \gamma \sqrt{\lambda_{max}(P)}R,\;\;\;\forall \;x\in \bS_r,\;\;\forall i\in \M.
\]
Thus, for all $x\in \bS_r$ and all $i\in \M$,
\begin{equation}\label{eq:CompleteInvariance}
\begin{aligned}
&\|A_ix+b_i\|_P \\& \leq \gamma \frac{R}{r}\frac{\sqrt{\lambda_{max}(P)}}{\sqrt{\lambda_{min}(P)}} r\sqrt{\lambda_{min}(P)}
\\&= \gamma \frac{R}{r}\frac{\sqrt{\lambda_{max}(P)}}{\sqrt{\lambda_{min}(P)}}\min_{y\in \bS_r}\|y\|_P\leq \gamma \frac{R}{r}\frac{\sqrt{\lambda_{max}(P)}}{\sqrt{\lambda_{min}(P)}}\|x\|_P.
\end{aligned}
\end{equation}
Defining $\wt \gamma\coloneqq\gamma \frac{R}{r}\frac{\sqrt{\lambda_{max}(P)}}{\sqrt{\lambda_{min}(P)}}$, we can thus say that
\begin{equation}\label{eq:BoundonSmallBall}
\|A_ix+b_i\|_P \leq \wt \gamma \|x\|_P,\;\;\;\forall x\in \bS_r,\;\;\forall \;i\in \M.
\end{equation}
This allows us, summarizing all the previous results, to give the following probabilistic upper bound of $\rho(\cA)$.
\zheming{
\begin{thm} \label{thm:SecondMethod}
Given any $\omega_N\in Z^N$ let us denote with  $(\gamma(\omega_N),P(\omega_N))$ the solution to $\mathcal{P}(\omega_N)$ defined in~(\ref{eqn:mathcalPomegaN}). Then, for any given $\beta \in(0,1)$,
\[
\mu^N\left \{\omega_N\in Z^N\,\,:\,\, \rho(\cA)\leq \bar\rho_2(\omega_N)\right \}\geq 1-\beta
\]
with 
\begin{equation}
    \bar\rho_2(\omega_N)= \frac{\gamma(\omega_N)\kappa(P(\omega_N))}{\delta(\varepsilon(s(\omega_N)) M)},
\end{equation}
where $s(\omega_N)$ is defined in (\ref{eqn:somegaN}),  $\varepsilon: \{0,1,\cdots, N\}\rightarrow [0,1]$ is a function defined as in (\ref{eqn:varepsilonk})
and $\kappa(P(\omega_N))\coloneqq\frac{\sqrt{\lambda_{max}(P(\omega_N))}}{\sqrt{\lambda_{min}(P(\omega_N))}}$.
Furthermore, if $\bar \rho_2(\omega_N)\leq 1$, the ellipsoid 
\[
E(\omega_N)\coloneqq\{x\in \R^n\;\vert\;x^\top P(\omega_N)x\leq \tilde L^2\}
\]
with $\tilde L =L\delta(\varepsilon(s(\omega_N)) M)$, is attractive and forward invariant with probability $1-\beta$.
\end{thm}
}
\begin{proof}
The first statement follows from~\eqref{eq:BoundonSmallBall} and Theorem~\ref{thm:scenario} and recalling Item 1) of Lemma~\ref{lemma:CoreLemma}. The forward invariance and attractiveness of $E(P\omega_N)$ follows from Item 3) of Lemma~\ref{lemma:CoreLemma}, applied to the norm $\|\cdot\|_{P(\omega_N)}:\R^n\to \R_{\geq 0}$.
\end{proof}
\begin{rem}
The radius $R>0$ of the ball whence the initial conditions are sampled is assumed to be given \textit{a priori} in this work. However, notice that there exists a trade-off in the choice of $R$; very large values can lead to a poor estimation of the forward invariant set whereas $R$ has to be large enough so that $K_\infty$ is contained in the interior of $\bS_R$. The task of determining such $R$ is far from trivial. Indeed, the presence of affine terms can spread $K_\infty$ over the whole state space and, intuitively, trajectories sampled on a sphere $\bS_R$ that evolve outwards can either indicate that the system is unstable or that $K_\infty$ exists but is not contained in the interior of $\bS_R$.  Therefore, it is reasonable to assume an upper bound $\overline R$ on the maximum radius from which one can sample initial conditions and, if for $\bS_{\overline R}$ the sampled problem~\eqref{eqn:mathcalPomegaN} fails to yield a feasible solution for a set of samples $\omega_N$, then item 2) of Lemma~\ref{lemma:CoreLemma} assures that it also fails for all $R<\overline R$ for some set of samples. On the other hand, if problem~\eqref{eqn:mathcalPomegaN} is successfully solved for $R=\overline R$, one can apply a bisection algorithm to find the smallest $R$ for which a solution to problem~\eqref{eqn:mathcalPomegaN} can be found, by drawing new data sets. 
\end{rem}~\\[.5cm]

\section{NUMERICAL EXPERIMENT}
To evaluate the methods for data-driven analysis of switched affine systems developed in the previous sections, we carried out numerical experiments considering academical examples. Let us consider two switched affine systems as in~\eqref{eq:Switchingsystems} denoted by $\cF_1$ and $\cF_2$ and given by
\begin{equation*}
    \mathcal{F}_1\!\!=\!\!\left\{\!\!\left(
    \begin{bmatrix}
      	0.4 & -0.3 \\ 
      	 -0.5 & 0.5
     \end{bmatrix}\!\!,\begin{bmatrix}
      	0.1 \\ 
      	 0.2
     \end{bmatrix}\right)\!\!,\left(
   \begin{bmatrix}
      	-0.3 & -0.1 \\ 
      	 -0.2 & -0.6
     \end{bmatrix}\!\!,\begin{bmatrix}
      	-0.2 \\ 
      	 -0.1
     \end{bmatrix} \right)\!\!
    \right\}
\end{equation*}
\begin{equation*}
    \mathcal{F}_2\!\!=\!\!\left\{\!\!\left(
    \begin{bmatrix}
        0.6 & 0.1 \\ 
       -0.2 & \!\!-0.5
     \end{bmatrix}\!\!,\begin{bmatrix}
      	-0.7 \\ 
      	 -0.7
     \end{bmatrix}\right)\!\!,\left(
   \begin{bmatrix}
      	-0.6 & \!\!-0.1 \\ 
      	 0.2 & 0.5
     \end{bmatrix}\!\!,\begin{bmatrix}
      	0.2 \\ 
      	 -0.8
     \end{bmatrix} \right)\!\!
    \right\}
\end{equation*}
For both systems $N=200$ initial conditions were uniformly sampled on the sphere $\bS_R$ with radius $R=3$ and the iteration defining the system~\eqref{eq:Switchingsystems} was applied once to each initial condition considering randomly sampled indexes $\sigma\in\M$. The obtained data formed the data set $\omega_N$ and the sampled optimization problem~\eqref{eqn:mathcalPomegaN} was successfully solved. Seeking to attain a confidence level of $1-\beta=0.95$,  we chose $\varepsilon=0.0882$ by solving ~\eqref{eqn:varepsilonkd} with $k=d$. Then, the proposed bounds $\bar\rho_1$ and $\bar\rho_2$ on the JSR, defined in Theorems~\ref{thm:first_bound} and~\ref{thm:SecondMethod}, were calculated, by considering $P$ provided by the second optimization problem~\eqref{eqn:second_step_optim_prob}. For the sake of repeatability, this procedure was performed 100 times for different random data sets to allow the evaluation of the average value and the standard deviation of both bounds. For $\cF_1$ we obtained $\bar\rho_1=0.9547\pm0.0065$ and $\bar\rho_2=1.0061\pm0.0070$ while for  $\cF_2$ the obtained bounds were
$\bar\rho_1 = 1.0273 \pm 0.0003$ and $\bar\rho_2=0.9876\pm0.0010$. This illustrates that either Theorem~\ref{thm:first_bound} and~\ref{thm:SecondMethod} can provide a tighter bound for different problems \lucas{and that one may fail to provide stability for cases where the other does not}. For each system, one of the sampled data sets is depicted in Figure~\ref{fig:state_space} along with the invariant set $E(\omega_N)$ ensured for $\mathcal{F}_2$ by Theorem~\ref{thm:SecondMethod}, as $\bar\rho_2<1$.
\begin{figure}
    \centering
    \includegraphics[width=0.49\linewidth]{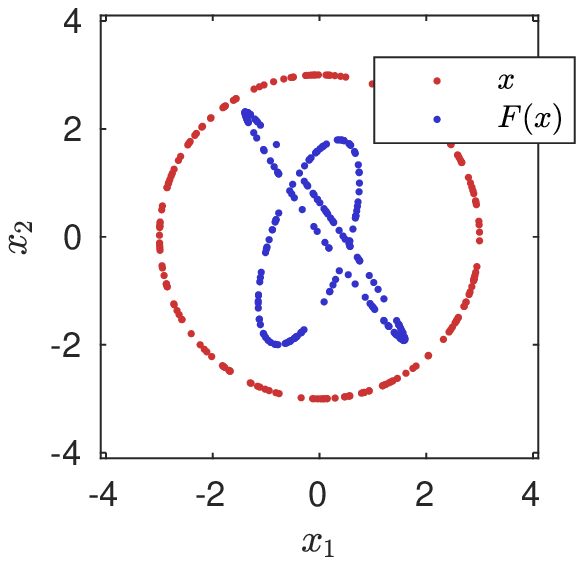}
    \includegraphics[width=0.49\linewidth]{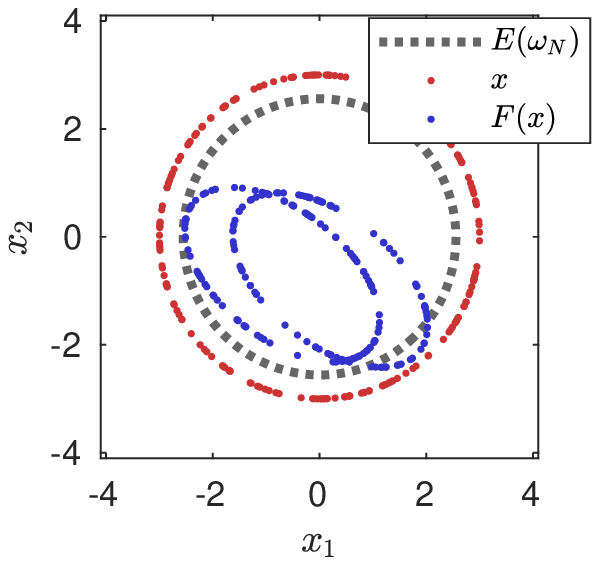}
    \caption{For systems $\mathcal{F}_1$ (left) and $\mathcal{F}_2$ (right),  $N=200$ points $x$ sampled from $\bS_R$, their evaluation of $F(x)$ for randomly chosen $\sigma\in \M$ and the corresponding invariant set $E(\omega_N)$, which is ensured only for $\mathcal{F}_2$ by Theorem~\ref{thm:SecondMethod} (given that $\bar\rho_2>1$ for $\mathcal{F}_1$ and an invariant set could not be ensured).}
    \label{fig:state_space}
\end{figure}

\section{Conclusion}
In this paper we have tackled the data-driven stability analysis problem for discrete-time affine switched systems. In this framework the subsystems composing the overall switched systems do not share a common equilibrium. However, we have shown that switched affine systems enjoy geometric properties, which allow to circumvent this problem. We have proposed two methods giving
probabilistic  certificates  of
stability: the first one necessitates an upper bound on the norm of the independent terms of the affine operators, while the second one does not, but may be less precise.

\appendix
\section{Joint Spectral Radius}\label{app:jsr}
\matteo{In this short section we recall the main definition of \emph{joint spectral radius} and some results used throughout the paper.
\begin{defn}
Consider a bounded set of matrices $\cA\subset \R^{n\times n}$. The \emph{joint spectral radius} (JSR) of $\cA$ is defined by
\begin{equation}
    \rho(\cA):=\lim_{l\to+\infty} \sup\{\|A\|^{\frac{1}{l}}\,\,\vert\,\,A\in \cA^l\}
\end{equation}
where $\cA^l\coloneqq\{A_{j_1}\cdot\cdot\cdot A_{j_l}\,\vert\,A_{j_i}\in \cA\}$ denotes the set all the products of length $l$ of matrices in $\cA$.
\end{defn}
Given $M\in \N$ and a set $\cA:=\{A_1,\dots, A_M\}\subset \R^{n\times n}$, it is well known that the \emph{linear} switching system $x(k+1)=A_{\sigma(k)}x(k)$ is globally asymptotically stable if and only if $\rho(\cA)<1$,  see e.g \cite[Corollary 1.1]{Jung09}.
Moreover, we recall the following result, for the proof see \cite[Proposition 1.4 and Section 2.1]{Jung09}.
\begin{lemma}\label{lemma:ExtremalNorms}
Given a bounded set of matrices $\cA\subset \R^{n\times n}$, for any $\wt \rho> \rho(\cA)$ there exists a \emph{$\wt\rho$-contractive norm} $\|\cdot\|:\R^n\to \R_{\geq 0}$ such that
\[
\|Ax\|\leq \wt \rho \|x\|,\,\,\,\forall\,x\in \R^n,\,\,\forall \,A\in \cA.
\]
Moreover, if the set $\cA$ is compact and irreducible, there exists an \emph{extremal norm} $\|\cdot\|_{\star}:\R^n\to \R_{\geq 0}$ such that
\[
\|Ax\|_{\star}\leq \rho(\cA) \|x\|_{\star},\,\,\,\forall\,x\in \R^n,\,\,\forall \,A\in \cA.
\]
\end{lemma}
} 
\section{Probabilistic Framework}\label{app:prob}
\zheming{
To be self-contained, we also review the probabilistic framework in \cite[Section 2]{KenBal19}. Consider the sphere $\bS_R$ in $\R^n$, the Borelian $\sigma$-algebra is denoted as $\cB_{\bS_R}$. Given any $S\in \cB_{\bS_R}$, $\bS^{S}:=\{tS: t\in [0,1]\}$ denotes the sector defined by $S$. We then equip $(\bS_R,\cB_{\bS_R})$ with the normalized uniform spherical measure, denoted by $\sigma^{n-1}$:
\begin{align}
    \sigma^{n-1}(S): = \frac{\lambda(\bS^S)}{\lambda(\B_R)},\;\;\;\; \forall S \in \cB_{\bS_R}
\end{align}
where $\B_R$ is the closed ball of radius $R$ in $\R^n$ and $\lambda(\cdot)$ denotes the Lebesgue measure. Similarly, we define the uniform measure on $\M$. Let $\Sigma_{\M}$ denote the $\sigma$-algebra of $\M$. The uniform measure on $(\M,\Sigma_{\M})$ is defined as:
\begin{align}
    \mu_M(s):= \frac{1}{M},\;\;\;\; \forall s\in \M.
\end{align}
The joint uniform measure on $(\bS_R \times \M, \cB_{\bS_R} \bigotimes \Sigma_{\M} )$ is then defined as: $\mu = \sigma^{n-1} \otimes \mu_M$. 

}

\end{document}